\documentclass[hidelinks,conference,hidelinks]{IEEEtran}

\usepackage{amsmath,amssymb,amsthm}
\usepackage{bm}
\usepackage{bbm}
\usepackage{orcidlink}
\usepackage{enumerate}
\usepackage{graphicx}
\usepackage{algorithmic}
\usepackage{algorithm}
\usepackage{cite}
\usepackage{hyperref}
\usepackage{xcolor}
\usepackage{booktabs}
\usepackage{array}
\usepackage{multirow}
\usepackage[acronym,shortcuts]{glossaries}
\usepackage{threeparttable}
\usepackage{cleveref}
\usepackage{balance}

\theoremstyle{definition}
\newtheorem{definition}{Definition}
\theoremstyle{remark}
\newtheorem{remark}{Remark}
\theoremstyle{plain}
\newtheorem{theorem}{Theorem}
\newtheorem{proposition}{Proposition}

\newcommand{\abs}[1]{|#1|}      
\newcommand{\norm}[1]{\|#1\|}   

\newcommand{\var}{\mathrm{Var}}
\newcommand{\figph}[1]{%
\IfFileExists{#1}{\includegraphics[width=\columnwidth]{#1}}%
{\framebox[0.9\columnwidth]{\rule{0pt}{3cm}\footnotesize\ttfamily\detokenize{#1}}}}

\newacronym{GaBP}{GaBP}{Gaussian belief propagation}
\newacronym{GAMP}{GAMP}{generalized approximate message passing}
\newacronym{AMP}{AMP}{approximate message passing}
\newacronym{CBM}{CBM}{continuous Bernoulli-Mises}
\newacronym{RDPC}{RDPC}{Radially-Discrete Phase-Continuous}
\newacronym{LE}{LE}{linear estimator}
\newacronym{MAP}{MAP}{maximum a posteriori}
\newacronym{MMSE}{MMSE}{minimum mean square error}
\newacronym{LMMSE}{LMMSE}{linear minimum mean square error}
\newacronym{SPA}{SPA}{sum-product algorithm}
\newacronym{CLT}{CLT}{central limit theorem}
\newacronym{LLN}{LLN}{law of large numbers}
\newacronym{SE}{SE}{state evolution}
\newacronym{MSE}{MSE}{mean square error}
\newacronym{m-SE}{m-SE}{mismatched state evolution}
\newacronym{EM}{EM}{expectation maximization}
\newacronym{ARD}{ARD}{automatic relevance determination}
\newacronym{CPT}{CPT}{continuous phase trap}
\newacronym{APSK}{APSK}{amplitude and phase-shift keying}
\newacronym{QAM}{QAM}{quadrature amplitude modulation}
\newacronym{KL}{KL}{Kullback-Leibler}
\newacronym{MIMO}{MIMO}{multiple-input multiple-output}
\newacronym{SNR}{SNR}{signal-to-noise ratio}
\newacronym{BER}{BER}{bit error rate}
\newacronym{UBAE}{UBAE}{universal blind amplitude equalizer}
\newacronym{MOS}{MOS}{model order selection}
\newacronym{PMF}{PMF}{probability mass function}
\newacronym{SRAM}{SRAM}{static random-access memory}
\newacronym{ASIC}{ASIC}{application-specific integrated circuit}
\newacronym{FPGA}{FPGA}{field programmable gate array}
\newacronym{ALU}{ALU}{arithmetic logic unit}
\newacronym{vMBP}{vMBP}{von Mises belief propagation}
\newacronym{PSK}{PSK}{phase-shift keying}
\newacronym{SotA}{SotA}{state-of-the-art}
\newacronym{2D}{2D}{two dimensional}
\newacronym{1D}{1D}{one dimensional}
\newacronym{AWGN}{AWGN}{additive white Gaussian noise}
\newacronym{VN}{VN}{variable node}
\newacronym{FN}{FN}{factor node}
\newacronym{ASG}{ASG}{asymptotic SNR gap}
\newacronym{iid}{i.i.d.}{independent and identically distributed}
\newacronym{OD-GaBP}{OD-GaBP}{orbital detection Gaussian belief propagation}
\newacronym{OD-GAMP}{OD-GAMP}{orbital detection generalized approximate message passing}
\newacronym{OD-AMP}{OD-AMP}{orbital detection approximate message passing}
\newacronym{OD}{OD}{orbital detection}
\newacronym{DCT}{DCT}{dominated convergence theorem}
\newacronym{GMI}{GMI}{generalized mutual information}
\newacronym{LPD}{LPD}{linear phase denoiser}
\newacronym{ML}{ML}{maximum likelihood}
\newacronym{CGR}{CGR}{continuous Gaussian relaxation}
\newacronym{I-MMSE}{I-MMSE}{information minimum mean square error}
\newacronym{ASK}{ASK}{amplitude-shift keying}
\newacronym{CRLB}{CRLB}{Cram\'{e}r-Rao lower bound}
\newacronym{MVUE}{MVUE}{minimum-variance unbiased estimator}
\newacronym{RS}{RS}{replica-symmetric}
\newacronym{EP}{EP}{expectation propagation}
\newacronym{EC}{EC}{expectation consistent}
\newacronym{GLM}{GLM}{generalized linear model}
\newacronym{SER}{SER}{symbol error rate}
\newacronym{LLR}{LLR}{log-likelihood ratio}
\newacronym{AEP}{AEP}{asymptotic equipartition property}
\newacronym{LDP}{LDP}{large deviation principle}
\newacronym{B-CBM}{$B$-CBM}{$B$-level continuous Bernoulli-Mises}
\newacronym{MUI}{MUI}{multi-user interference}
\newacronym{OAMP}{OAMP}{orthogonal approximate message passing}
\newacronym{MAMP}{MAMP}{memory approximate message passing}
\newacronym{PAM}{PAM}{pulse amplitude modulation}

\newcommand\blfootnote[1]{%
  \begingroup
  \renewcommand\thefootnote{}%
  \footnote{#1}%
  \addtocounter{footnote}{-1}%
  \endgroup
}

\begin{document}

\title{Orbital Detection: On Maximum-Entropy Priors}

\author{\IEEEauthorblockN{Kuranage~Roche~Rayan~Ranasinghe$^\star$, Takumi Takahashi$^\dag$, and Giuseppe Thadeu Freitas de Abreu$^\star$}\\[-2.5ex]
\IEEEauthorblockA{\textit{$^\star$School of Computer Science and Engineering, Constructor University, Bremen, Germany} \\
\textit{$^\dag$Graduate School of Engineering, Osaka University, Suita, Japan} \\
Emails: \{kranasinghe, gabreu\}@constructor.university, takahashi@comm.eng.osaka-u.ac.jp}
\vspace{-4ex}}

\maketitle

\begin{abstract}
Soft-input detection over a discrete constellation $\mathcal{M}$ of cardinality $M$ requires computing a posterior whose mean and mode are respectively given by the \ac{MMSE} and \ac{MAP} estimates, both of which incur a computational cost of order $\mathcal{O}(M)$ per symbol.
We show that this cost is reduced\footnote{Continuous ring relaxations of this kind have appeared before; e.g. \cite{Tanahashi2011,Hara2026,Suresh2026}, as heuristic approximations of the symbol distribution, and the associated Bessel-ratio denoiser is also known, such that neither is claimed here.
Our contribution is therefore to show that the construction is not a heuristic.} to $\mathcal{O}(L)$, where $L \le M$ is the number of distinct amplitudes (rings), once the discrete prior is replaced by its maximum-entropy counterpart subject to the same radial marginal. 
This \emph{orbital prior}, which is a mixture of uniform circular shells, is obtained by maximizing a mixed discrete-continuous entropy.
We prove in this paper that such a distribution is the only distribution on $\mathbb{C}$ that preserves the amplitude statistics of $\mathcal{M}$ exactly while remaining maximally noncommittal in phase. 
Under the \ac{AWGN} channel, the orbital prior induces a closed-form posterior that factors into a softmax over the $L$ rings and a von Mises phase distribution whose concentration is supplied entirely by the observation, yielding closed-form orbital \ac{MMSE} and \ac{MAP} detectors of the discrete symbol at $\mathcal{O}(L)$ cost.
The resulting hierarchical rule selects the ring by posterior \emph{mass} and the phase by conditional \emph{mode}.
We compare the pairwise ring boundary with that of the joint posterior-\emph{density} and quantify the leading-order outward shift at high \ac{SNR}.
Numerical results using standard constellations confirm that the orbital detectors maintain similar \ac{SER} performance to optimal detectors, at a fraction of the complexity.
\end{abstract}

\vspace{-2ex}
\begin{IEEEkeywords}
Orbital detection, maximum-entropy, directional statistics, posterior computation, soft-input detection.
\end{IEEEkeywords}

\glsresetall

\blfootnote{A complete journal article detailing \ac{OD} has been submitted to the Transactions on Information Theory~\cite{ranasinghe2026orbital}.}

\vspace{-3ex}
\section{Introduction}
\label{sec:intro}
Detection over a discrete constellation is the computation of a posterior. 
Given a noisy observation, the \ac{MMSE} estimate is the posterior mean and the \ac{MAP} estimate is its mode. Each of these weighs or compares all $M$ constellation points, at a cost of $\mathcal{O}(M)$ per symbol. 
In modern iterative receivers built on \ac{AMP}, \ac{GAMP}, and \ac{GaBP}~\cite{Donoho2009,Bayati2011,Rangan2011,ranganVAMP2019,ma2017orthogonal,Cespedes2014}, this per-symbol \emph{denoiser}, which computes the posterior mean and variance from an effective scalar observation, is evaluated millions of times. 
Its cost is therefore the dominant bottleneck, while the posterior information it produces drives the iteration toward its fixed point~\cite{Rangan2011,Guo2005,GuoWuShamaiVerdu_TIT_2011}.

The standard inexpensive surrogate replaces the discrete prior by an energy-matched complex Gaussian, yielding the linear \ac{MMSE} (Wiener) detector. 
This is cheap, but it discards \emph{all} structure of $\mathcal{M}$ and therefore stalls iterative detection well short of the discrete fixed point~\cite{tse2005fundamentals,proakis2007digital}. 
A gap therefore separates the crude Gaussian from the exact discrete distribution, the evaluation of which costs $\mathcal{O}(M)$.

The question we address in this article is whether a principled relaxation exists which retains the part of $\mathcal{M}$ that matters while remaining cheap to evaluate.
Relaxations of this kind are not new.
For conditional-mean \ac{MIMO} detection, Tanahashi and Ochiai~\cite{Tanahashi2011} replace the discrete symbol distribution by a continuous one, their uniform ring approximation for \ac{PSK} and \ac{APSK} spreading infinitely many points over each amplitude ring, so that the resulting weighted sum of Dirac shells coincides with the prior studied here and integrates against the Gaussian likelihood in closed form via the modified Bessel functions $I_0$ and $I_1$.
In a different context, Hara \emph{et al.}~\cite{Hara2026} reach the same $I_1/I_0$ denoiser on a single ring as the zero-concentration, non-informative case of a von Mises prior for reciprocity calibration, and related constant-modulus constructions have been employed for massive \ac{MIMO} detection~\cite{Suresh2026}.
We therefore claim neither the continuous ring relaxation nor the scalar Bessel-ratio denoiser as contributions of this work.
What has been missing is a justification of \emph{why} this particular relaxation is the right choice, and whether it is the only distribution on $\mathbb{C}$ preserving the amplitude statistics of $\mathcal{M}$ while remaining maximally flexible in phase: the prior art posits the ring density as a convenient approximation, close to the constellation and integrable in closed form, and then discards the Bessel expressions in favor of asymptotic surrogates for implementation~\cite{Tanahashi2011}.
We show instead that it needs no such justification by convenience, being the unique answer to a variational problem.

We answer this question through the lens of maximum entropy~\cite{Jaynes1957}. 
Retaining only the \emph{radial marginal} of the constellation, that is, the distribution of its amplitude, and maximizing a mixed discrete-continuous entropy~\cite{Renyi1959,WuVerd2010}, we obtain the \emph{orbital prior}, a mixture of uniform circular shells, and prove that it is the unique distribution on $\mathbb{C}$ that preserves the amplitude statistics of $\mathcal{M}$ exactly while remaining maximally noncommittal in phase. 

Under \ac{AWGN}, this prior induces a closed-form posterior whose phase distribution is von Mises~\cite{mardia2009directional} with concentration supplied by the observation, and whose ring weights form a softmax over $L \le M$ amplitude levels. 
Consequently, the orbital \ac{MMSE} and \ac{MAP} detectors are closed-form and cost only $\mathcal{O}(L)$ per symbol.

Our contributions are as follows.
Section~\ref{sec:orbital_prior} derives, by the chain rule, the mixed discrete-continuous entropy whose unique maximizer is the orbital prior, characterized thereby as the only amplitude-preserving, phase-noncommittal relaxation of $\mathcal{M}$ (Theorem~\ref{thm:maxent}).
Section~\ref{sec:orbital_posterior} derives the closed-form orbital posterior and the resulting $\mathcal{O}(L)$ \ac{MMSE} and \ac{MAP} detectors, retaining the Bessel functions exactly and obtaining a hierarchical \ac{MAP} rule with no counterpart in earlier ring constructions, together with the exact pairwise ring threshold of the joint density mode, the high-\ac{SNR} expansion of the posterior-mass threshold, and their strict ordering whenever the density threshold is positive (Proposition~\ref{prop:ring_boundary}).
Finally, Section~\ref{sec:numerical} offers numerical evidence that \ac{OD} matches the exact discrete \ac{SER} to within a small high-\ac{SNR} gap, at $\mathcal{O}(L)$ cost.

\emph{Notation:} Random variables and their realizations are respectively denoted by sans-serif and italic letters; e.g., ($\mathsf{x},\mathsf{y},\mathsf{z}$) and ($x,y,z$), with tilded counterparts $\tilde{\mathsf{x}},\tilde{x}$ for the relaxed quantities under the orbital approximation, and calligraphic letters (e.g. $\mathcal{M}$) for sets, $\abs{\mathcal{A}}$ being the cardinality of $\mathcal{A}$.
For $z \in \mathbb{C}$, $\norm{z}$ and $\angle z$ denote its modulus and phase, $j \triangleq \sqrt{-1}$, and $\mathcal{CN}(\mu,\sigma^{2})$ is the circularly symmetric complex Gaussian distribution.
We write $\Pr(\cdot)$, $p_{\,\cdot}(\cdot)$, $\mathbb{E}[\cdot]$ and $\var[\cdot]$ for probability, density, expectation and variance, and abbreviate conditioning on $\{\mathsf{y}=y\}$ under the true channel as $\mid y$, and on $\{\tilde{\mathsf{y}}=\tilde{y}\}$ under the orbital channel of Section~\ref{sec:orbital_posterior} as $\mid \tilde{y}$.
Here, $H(\cdot)$ and $h(\cdot)$ are the discrete and differential entropies (in nats), $I_{\nu}(\cdot)$ the modified Bessel function of the first kind of order $\nu$~\cite{Abramowitz1965,DLMF}, $\mathrm{vM}(\theta;\mu,\kappa)$ the von Mises density of mean direction $\mu$ and concentration $\kappa$, $\delta(\cdot)$ the Dirac delta, $\triangleq$ equality by definition, and $\mathcal{O}(\cdot)$ the Landau symbol.

\vspace{-0.5ex}
\section{Foundations of Orbital Detection}
\label{sec:system}
\vspace{-0.5ex}

A realization $x$ of a discrete symbol $\mathsf{x}$ is observed through a scalar \ac{AWGN} channel as
\vspace{-1ex}
\begin{equation}
\label{eq:system_model}
\vspace{-1ex}
y = x + z \in \mathbb{C},
\end{equation}
where $z \sim \mathcal{CN}(0,\sigma_{\mathsf{z}}^2)$ is \ac{AWGN} with variance $\sigma_{\mathsf{z}}^2$, and $\mathsf{x}$ is assumed to be drawn from a discrete constellation
\vspace{-1ex}
\begin{equation}
\label{eq:constellation}
\vspace{-1ex}
\mathcal{M} \triangleq \{s_1, s_2, \ldots, s_M\} \subset \mathbb{C},
\end{equation}
of cardinality $M \triangleq \abs{\mathcal{M}}$ with prior $p_m \triangleq \Pr(\mathsf{x} = s_m)$ (trivially, we have $p_m = 1/M$ for equiprobable signaling). 


The constellation is zero-mean, so its average power equals the variance of $\mathsf{x}$, given by
\vspace{-1ex}
\begin{equation}
\label{eq:sigma_x}
\vspace{-1ex}
\sigma_{\mathsf{x}}^2 \triangleq \var[\mathsf{x}] = \mathbb{E}\big[\norm{\mathsf{x}}^2\big] = \sum_{m=1}^{M} p_m \norm{s_m}^2 .
\end{equation}

Bayes' rule then yields the posterior weights
\vspace{-1ex}
\begin{equation}
\label{eq:exact_posterior}
\vspace{-1ex}
\hat{p}_m \triangleq \Pr(\mathsf{x} = s_m \mid y)
     = \frac{p_m \exp\Big(\frac{-\norm{y - s_m}^2}{\sigma_{\mathsf{z}}^2}\Big)}
     {\sum_{m'=1}^{M} p_{m'} \exp\Big(\frac{-\norm{y - s_{m'}}^2}{\sigma_{\mathsf{z}}^2}\Big)},
\end{equation}
whose mean and mode are the \ac{MMSE} and \ac{MAP} estimates, respectively, given by
\vspace{-1ex}
\begin{align}
\label{eq:mmse_exact}
\hat{x}_{\mathrm{MMSE}} &= \mathbb{E}[\mathsf{x} \mid y] = \sum_{m=1}^{M} s_m\, \hat{p}_m, \\
\label{eq:map_exact}
\hat{x}_{\mathrm{MAP}} &= \arg\max_{s_m \in \mathcal{M}} \;\hat{p}_m.
\vspace{-1ex}
\end{align}

Each computes the same posterior at a computational cost of $\mathcal{O}(M)$ per observation, since the \ac{MMSE} mean weights all $M$ points and the \ac{MAP} mode compares all $M$ of them~\cite{YangHanzo2015}.
Reducing this cost is the fundamental objective of this paper.

The geometry of $\mathcal{M}$ is carried by its amplitudes. 
Let $R_1 < R_2 < \cdots < R_L$ be the $L$ distinct values of $\norm{s_m}$ over $s_m \in \mathcal{M}$, which partition the constellation into $L$ rings as
\vspace{-1ex}
\begin{equation}
\label{eq:ring_partition}
\vspace{-1ex}
\mathcal{M}_\ell \triangleq \{\, s_m \in \mathcal{M} \mid \norm{s_m} = R_\ell \,\},
\qquad \ell = 1, \ldots, L,
\end{equation}
which are disjoint and exhaustive. 

Ring $\ell$ holds $M_\ell \triangleq \abs{\mathcal{M}_\ell}$ symbols, with $\sum_{\ell} M_\ell = M$, and carries the \emph{ring prior}
\vspace{-1ex}
\begin{equation}
\label{eq:ring_prior}
r_\ell \triangleq \Pr(\norm{\mathsf{x}} = R_\ell) = \sum_{m \,\mid\, s_m \in \mathcal{M}_\ell} p_m,
\vspace{-1ex}
\end{equation}
which reduces to $r_\ell = M_\ell/M$ for equiprobable signaling. 

We call $\{(R_\ell, r_\ell)\}_{\ell=1}^{L}$ the \emph{radial marginal} of $\mathsf{x}$. 
It is everything the geometry of $\mathcal{M}$ records once phase is ignored, and it is the only structure Section~\ref{sec:orbital_prior} retains in building the least-committal prior. 
Thus, $\mathcal{M}$ supports two descriptions, namely, the full discrete probability space $\{(s_m, p_m)\}_{m=1}^{M}$, which carries phase, and its radial marginal $\{(R_\ell, r_\ell)\}_{\ell=1}^{L}$, which does not.
%

\vspace{-1ex}
\section{The Orbital Prior from Maximum Entropy}
\label{sec:orbital_prior}
We now build the least-committal prior that retains the exact radial marginal $\{(R_\ell, r_\ell)\}$ and nothing more. 
Under the discrete probability space $\{(s_m,p_m)\}$ both the modulus and the phase are discrete, and the chain rule for discrete entropy~\cite[Theorem 2.2.1]{CoverThomas2006} gives
\vspace{-1ex}
\begin{equation}
\label{eq:discrete_entropy}
\vspace{-1ex}
H(\mathsf{x}) = H(\norm{\mathsf{x}}, \angle \mathsf{x}) = H(\norm{\mathsf{x}}) + H(\angle \mathsf{x} \mid \norm{\mathsf{x}}).
\end{equation}

Retaining only $\{(R_\ell, r_\ell)\}$ fixes the radial term, namely, $H(\norm{\mathsf{x}}) = -\sum_\ell r_\ell \ln r_\ell$, thus freeing the phase to be relaxed, producing a new process on $\mathbb{C}$, which is no longer the discrete symbol $\mathsf{x}$ and which we denote by $\tilde{\mathsf{x}}$. This process shares the radial marginal, $\norm{\tilde{\mathsf{x}}} \stackrel{d}{=} \norm{\mathsf{x}}$, but carries a continuous phase via a conditional density $p_{\angle\tilde{\mathsf{x}} \mid \norm{\tilde{\mathsf{x}}}}(\theta \mid R_\ell)$, denoted $q_\ell(\theta)$ on each ring. 
For $\tilde{\mathsf{x}}$, the discrete entropy $H(\angle \mathsf{x} \mid \norm{\mathsf{x}})$ of~\eqref{eq:discrete_entropy} is replaced by the conditional \emph{differential} entropy~\cite[Theorem 9.6.2]{CoverThomas2006}
\vspace{-1ex}
\begin{align}
\label{eq:cond_diff_entropy}
\vspace{-1ex}
h(\angle\tilde{\mathsf{x}} \mid \norm{\tilde{\mathsf{x}}}) &\triangleq \sum_{\ell=1}^{L} r_\ell\, h(\angle\tilde{\mathsf{x}} \mid \norm{\tilde{\mathsf{x}}} = R_\ell) \nonumber \\
&= -\sum_{\ell=1}^{L} r_\ell \int_{-\pi}^{\pi} q_\ell(\theta) \ln q_\ell(\theta)\, \mathrm{d}\theta,
\vspace{-1ex}
\end{align}
yielding the \emph{mixed entropy}
\vspace{-1ex}
\begin{align}
\label{eq:entropy_chain}
J(\tilde{\mathsf{x}}) &\triangleq H(\norm{\tilde{\mathsf{x}}}) + h(\angle\tilde{\mathsf{x}} \mid \norm{\tilde{\mathsf{x}}})  \\
&= \underbrace{-\sum_{\ell=1}^{L} r_\ell \ln r_\ell}_{\text{fixed by }\{(R_\ell,r_\ell)\}}
\;+\; \underbrace{\sum_{\ell=1}^{L} r_\ell\, h\big(\angle\tilde{\mathsf{x}} \mid \norm{\tilde{\mathsf{x}}} = R_\ell\big)}_{\text{free}} . \nonumber
\end{align}

Formally, $J(\tilde{\mathsf{x}})$ is the entropy of $\tilde{\mathsf{x}}$ relative to the \emph{mixed reference measure} given by counting on $\{R_\ell\}$ times Lebesgue on $[-\pi,\pi)$, and is thus a well-defined mixed discrete-continuous entropy~\cite{Renyi1959,WuVerd2010}. 
Taken against this mixed measure rather than against the counting measure, $J(\tilde{\mathsf{x}})$ is \emph{not} comparable in value to the discrete entropy of~\eqref{eq:discrete_entropy}, because it is the entropy of a different random variable against a different reference. Only its structure as a functional of the phase densities $\{q_\ell\}$ matters here.
As $H(\norm{\tilde{\mathsf{x}}})$ is fixed by the radial marginal, maximizing $J(\tilde{\mathsf{x}})$ reduces to maximizing each ring's phase entropy independently.
Since any phase that is not absolutely continuous has $h(\angle\tilde{\mathsf{x}} \mid \norm{\tilde{\mathsf{x}}} = R_\ell) = -\infty$, the maximizer admits a density $q_\ell(\theta)$, so this relaxation is without loss.

On ring $\ell$, the support $[-\pi,\pi)$ is bounded, so normalization $\int_{-\pi}^{\pi} q_\ell(\theta)\,\mathrm{d}\theta = 1$ is the only constraint needed for a proper maximizer. 
Introducing a multiplier $\lambda_\ell$, the Lagrangian is~\cite{Jaynes1957}
\vspace{-1ex}
\begin{equation}
\label{eq:lagrangian}
\vspace{-1ex}
\mathcal{L}[q_\ell, \lambda_\ell]
= -\!\int_{-\pi}^{\pi}\! q_\ell \ln q_\ell\, \mathrm{d}\theta
+ \lambda_\ell\!\left( \int_{-\pi}^{\pi}\! q_\ell\, \mathrm{d}\theta - 1 \right),
\end{equation}
and has the stationarity condition
\vspace{-0.5ex}
\begin{equation}
\label{eq:euler_lagrange}
\vspace{-0.5ex}
\frac{\delta \mathcal{L}}{\delta q_\ell(\theta)} = -\ln q_\ell(\theta) - 1 + \lambda_\ell = 0,
\vspace{-1ex}
\end{equation}
yielding $q_\ell(\theta) = e^{\lambda_\ell - 1}$, which is a constant in $\theta$. 

Normalizing, and using the strict concavity of the entropy for uniqueness gives
\vspace{-1.5ex}
\begin{equation}
\label{eq:uniform_phase}
q_\ell(\theta) = \frac{1}{2\pi}.
\end{equation}

In polar coordinates, $\tilde{x} = \rho\, e^{j\theta}$ with $\rho \equiv \norm{\tilde{x}}$ and $\theta \equiv \angle\tilde{x}$, and the area element is $\mathrm{d}\tilde{x} = \rho\,\mathrm{d}\rho\,\mathrm{d}\theta$.
Ring $\ell$ is confined to $\norm{\tilde{x}} = R_\ell$, so its density relative to $\mathrm{d}\tilde{x}$ has the form $p^{\mathrm o}_\ell(\tilde{x}) = \gamma_\ell(\theta)\,\delta(\norm{\tilde{x}} - R_\ell)$ for some $\gamma_\ell(\theta)$.
Marginalizing the radial direction (where $\norm{\tilde{x}}=\rho$) must reproduce $q_\ell$, giving
\begin{equation}
\label{eq:induced_phase}
\int_0^\infty \gamma_\ell(\theta)\,\delta(\rho - R_\ell)\, \rho\, \mathrm{d}\rho = \gamma_\ell(\theta)\, R_\ell \;\overset{!}{=}\; q_\ell(\theta).
\end{equation}

Therefore, $\gamma_\ell(\theta) = q_\ell(\theta)/R_\ell$ and, using \eqref{eq:uniform_phase}, we have
\begin{equation}
\label{eq:ring_density_final}
p^{\mathrm o}_\ell(\tilde{x}) = \frac{q_\ell(\theta)}{R_\ell}\,\delta(\norm{\tilde{x}} - R_\ell) = \frac{\delta(\norm{\tilde{x}} - R_\ell)}{2\pi R_\ell}.
\end{equation}

We are now ready to formally define the orbital prior.

\begin{definition}[Orbital prior]
\label{def:orbital_prior}
Relative to the polar measure $\mathrm{d}\tilde{x} = \rho\,\mathrm{d}\rho\,\mathrm{d}\theta$, the \emph{orbital prior} is the distribution of $\tilde{\mathsf{x}}$ with density
\vspace{-1ex}
\begin{equation}
\label{eq:orbital_prior}
p^{\mathrm o}(\tilde{x}) \triangleq \sum_{\ell=1}^{L} r_\ell\,p^{\mathrm o}_\ell(\tilde{x})
= \sum_{\ell=1}^{L} r_\ell\, \frac{\delta(\norm{\tilde{x}} - R_\ell)}{2\pi R_\ell},
\end{equation}
a mixture of uniform circular shells.

It is a probability distribution, $\int_{\mathbb{C}} p^{\mathrm o}\, \mathrm{d}\tilde{x} = \sum_{\ell} r_\ell = 1$, and $\tilde{\mathsf{x}}$ shares the same mean and average energy as $\mathsf{x}$: $\mathbb{E}[\tilde{\mathsf{x}}] = 0 = \mathbb{E}[\mathsf{x}]$ and $\mathbb{E}[\norm{\tilde{\mathsf{x}}}^2] = \sum_{\ell} r_\ell R_\ell^2 = \sigma_{\mathsf{x}}^2 = \mathbb{E}[\norm{\mathsf{x}}^2]$.
\end{definition}


\begin{theorem}[Maximum-Entropy Characterization]
\label{thm:maxent}
Among all distributions on $\mathbb{C}$ with radial marginal $\{(R_\ell, r_\ell)\}_{\ell=1}^{L}$, the orbital prior $p^{\mathrm o}$ uniquely maximizes the mixed entropy $J(\tilde{\mathsf{x}})$.
\end{theorem}
\begin{proof}
Any phase conditional that is not absolutely continuous (w.r.t.\ Lebesgue on $[-\pi,\pi)$) has $h(\angle\tilde{\mathsf{x}} \mid \norm{\tilde{\mathsf{x}}} = R_\ell) = -\infty$ by the standard extension of differential entropy to singular measures~\cite[\S 8]{CoverThomas2006}, so by \eqref{eq:entropy_chain}, $J(\tilde{\mathsf{x}}) = -\infty$ and cannot maximize. The maximizer therefore admits a density $q_\ell(\theta)$.
By \eqref{eq:entropy_chain}, the radial entropy $H(\norm{\tilde{\mathsf{x}}})$ is fixed, so $J(\tilde{\mathsf{x}})$ is maximized iff each phase entropy $h(\angle\tilde{\mathsf{x}} \mid \norm{\tilde{\mathsf{x}}} = R_\ell)$ is.
By \eqref{eq:euler_lagrange}--\eqref{eq:uniform_phase} and the strict concavity of the entropy, each is uniquely maximized by $1/(2\pi)$.
The unique maximizer therefore has uniform phase on every ring, which is the orbital prior $p^{\mathrm o}$.
\end{proof}

\vspace{-2ex}
\section{The Orbital Posterior}
\label{sec:orbital_posterior}
Definition~\ref{def:orbital_prior} fixes the prior of $\tilde{\mathsf{x}}$, and we now compute the exact posterior it induces under the \emph{orbital channel}
\vspace{-1ex}
\begin{equation}
\label{eq:orbital_channel}
\vspace{-1ex}
\tilde{\mathsf{y}} \;=\; \tilde{\mathsf{x}} + \mathsf{z},
\end{equation}
which shares the noise $\mathsf{z}\sim\mathcal{CN}(0,\sigma_{\mathsf{z}}^2)$ of~\eqref{eq:system_model}. 

Realizations of $\tilde{\mathsf{y}}$ are denoted $\tilde{y}$, mirroring $y$ for $\mathsf{y}$. The orbital detector is the resulting function of $\tilde{y}$. It is applied to data from the true channel by evaluating at $\tilde{y} = y$, and this substitution is the framework's sole approximation.
Since the orbital prior factors into a ring distribution $\{r_\ell\}$ and a uniform phase on each ring, the posterior of $\tilde{\mathsf{x}}$ inherits the same product form, with the channel supplying the phase concentration that the prior withheld.
The orbital \ac{AWGN} likelihood is
\vspace{-1ex}
\begin{equation}
\label{eq:AWGN_likelihood}
\vspace{-1ex}
p_{\tilde{\mathsf{y}} \mid \tilde{\mathsf{x}}}(\tilde{y} \mid \tilde{x}) = \frac{1}{\pi \sigma_{\mathsf{z}}^2} \exp\big(\tfrac{-\norm{\tilde{y} - \tilde{x}}^2}{\sigma_{\mathsf{z}}^2}\big),
\end{equation}
which is identical in functional form to $p_{\mathsf{y}\mid \mathsf{x}}$ since $\tilde{\mathsf{y}}-\tilde{\mathsf{x}}=\mathsf{z}=\mathsf{y}-\mathsf{x}$ in distribution.

Expressing $\tilde{x} = R_\ell e^{j\theta}$ and $\tilde{y} = \norm{\tilde{y}} e^{j \angle \tilde{y}}$, the squared distance expands as $\norm{\tilde{y} - \tilde{x}}^2 = \norm{\tilde{y}}^2 + R_\ell^2 - 2 R_\ell \norm{\tilde{y}} \cos(\theta - \angle \tilde{y})$, such that
\vspace{-2ex}
\begin{equation}
\label{eq:likelihood_polar}
p_{\tilde{\mathsf{y}} \mid \tilde{\mathsf{x}}}\big(\tilde{y} \mid R_\ell e^{j\theta}\big)
= \frac{e^{-(\norm{\tilde{y}}^2 + R_\ell^2)/\sigma_{\mathsf{z}}^2}}{\pi \sigma_{\mathsf{z}}^2}\,
\overbrace{e^{\kappa_\ell \cos(\theta - \angle \tilde{y})}}^{\text{von Mises kernel}},
\end{equation}
with concentration
\vspace{-0.5ex}
\begin{equation}
\label{eq:kappa}
\kappa_\ell \triangleq \frac{2 R_\ell \norm{\tilde{y}}}{\sigma_{\mathsf{z}}^2}.
\end{equation}

The von Mises kernel is not posited but produced by the Gaussian quadratic in polar coordinates, and its concentration is fixed by the channel. 
We remark that the same kernel is obtained in~\cite{Hara2026} as the posterior arising from a von Mises prior of concentration $\beta$, of which~\eqref{eq:likelihood_polar} is the non-informative case $\beta=0$ singled out by Theorem~\ref{thm:maxent}.

Marginalizing the phase against the uniform prior $1/(2\pi)$ via $\int_{-\pi}^{\pi} e^{\kappa \cos\beta}\, \mathrm{d}\beta = 2\pi I_0(\kappa)$ yields the ring marginal 
\vspace{-1ex}
\begin{equation}
\label{eq:ring_marginal}
\vspace{-1ex}
p_{\tilde{\mathsf{y}} \mid \norm{\tilde{\mathsf{x}}}}(\tilde{y} \mid R_\ell)
= \frac{1}{\pi \sigma_{\mathsf{z}}^2}\, e^{-(\norm{\tilde{y}}^2 + R_\ell^2)/\sigma_{\mathsf{z}}^2}\, I_0(\kappa_\ell).
\end{equation}

By Bayes' rule, the ring posterior $\hat{r}_\ell \triangleq \Pr(\norm{\tilde{\mathsf{x}}} = R_\ell \mid \tilde{y})$ is a softmax over ring log-metrics, and discarding the $\ell$-independent factors yields
\vspace{-1ex}
\begin{equation}
\label{eq:ring_logmetric}
\vspace{-1ex}
\Lambda_\ell \triangleq \ln r_\ell - \frac{R_\ell^2}{\sigma_{\mathsf{z}}^2} + \ln I_0(\kappa_\ell),\; \text{with}\;
\hat{r}_\ell = \frac{e^{\Lambda_\ell}}{\sum_{\ell'=1}^{L} e^{\Lambda_{\ell'}}},
\end{equation}
while the phase, conditioned on a ring, is the von Mises distribution
\vspace{-1ex}
\begin{equation}
\label{eq:vm_phase}
\vspace{-1ex}
p_{\angle\tilde{\mathsf{x}} \mid \norm{\tilde{\mathsf{x}}}, \tilde{\mathsf{y}}}(\theta \mid R_\ell, \tilde{y})
= \mathrm{vM}\big(\theta;\, \angle \tilde{y},\, \kappa_\ell\big)
= \frac{e^{\kappa_\ell \cos(\theta - \angle \tilde{y})}}{2\pi I_0(\kappa_\ell)}.
\end{equation}

Collecting these, the orbital posterior factorizes as
\begin{equation}
\label{eq:orbital_posterior}
p_{\norm{\tilde{\mathsf{x}}}, \angle\tilde{\mathsf{x}} \mid \tilde{\mathsf{y}}}(R_\ell, \theta \mid \tilde{y})
= \hat{r}_\ell\, \mathrm{vM}\big(\theta;\, \angle \tilde{y},\, \kappa_\ell\big),
\end{equation}
mirroring the prior of Definition~\ref{def:orbital_prior}, in the sense that the channel converts the uniform prior phase into a von Mises posterior phase whose concentration $\kappa_\ell$ comes entirely from the observation. 

The von Mises mean resultant denotes the Bessel ratio $A(\kappa) \triangleq I_1(\kappa)/I_0(\kappa)$~\cite{Watson1944}. 
Substituting $\beta = \theta - \angle \tilde{y}$ and using $\int_{-\pi}^{\pi} \cos\beta\, e^{\kappa \cos\beta}\, \mathrm{d}\beta = 2\pi I_1(\kappa)$, yields
\vspace{-1ex}
\begin{equation}
\label{eq:bessel_ratio}
\vspace{-1ex}
\mathbb{E}\big[\tilde{\mathsf{x}} \mid \norm{\tilde{\mathsf{x}}} = R_\ell,\, \tilde{y}\big] = R_\ell\, A(\kappa_\ell)\, e^{j \angle \tilde{y}}.
\end{equation}

\vspace{-1ex}
\begin{proposition}[Orbital detectors]
\label{prop:orbital_estimators}
Under the orbital prior, the orbital \ac{MMSE} and \ac{MAP} estimators of $\mathsf{x}$ from the realized $y$ are obtained by evaluating the orbital posterior~\eqref{eq:orbital_posterior} at $\tilde{y}=y$:
\vspace{-3ex}
\begin{align}
\label{eq:mmse_estimate}
\vspace{-1ex}
\hat{x}_{\mathrm{OMMSE}}(y) &\;\triangleq\; \mathbb{E}[\tilde{\mathsf{x}}\mid\tilde{\mathsf{y}}=y] \;=\; e^{j \angle y} \sum_{\ell=1}^{L} \hat{r}_\ell\, R_\ell\, A(\kappa_\ell), \\
\label{eq:map_estimate}
\hat{x}_{\mathrm{OMAP}}(y) &\;\triangleq\; R_{\ell^\star}\, e^{j \angle y},
\qquad \ell^\star = \arg\max_{\ell}\, \Lambda_\ell,
\end{align}
with $\kappa_\ell, \Lambda_\ell, \hat{r}_\ell$ as in~\eqref{eq:kappa}--\eqref{eq:ring_logmetric} taken at $\tilde{y}=y$ (so $\norm{\tilde{y}}=\norm{y}$ and $\angle\tilde{y}=\angle y$). The associated orbital posterior variance is
\vspace{-1.5ex}
\begin{equation}
\label{eq:posterior_var}
\vspace{-1ex}
\var[\tilde{\mathsf{x}}\mid\tilde{\mathsf{y}}=y] \;=\; \sum_{\ell=1}^{L} \hat{r}_\ell\, R_\ell^2 - \norm{\hat{x}_{\mathrm{OMMSE}}(y)}^2.
\end{equation}
\end{proposition}
\vspace{-1ex}
\begin{proof}
Under the orbital model, the conditional mean is $\mathbb{E}[\tilde{\mathsf{x}}\mid\tilde{y}] = \sum_\ell\hat{r}_\ell\,\mathbb{E}[\tilde{\mathsf{x}}\mid\norm{\tilde{\mathsf{x}}}=R_\ell,\tilde{y}]$, which by~\eqref{eq:bessel_ratio} equals $e^{j\angle\tilde{y}}\sum_\ell\hat{r}_\ell R_\ell A(\kappa_\ell)$. The variance follows from $\mathbb{E}[\norm{\tilde{\mathsf{x}}}^2\mid\tilde{y}]=\sum_\ell\hat{r}_\ell R_\ell^2$ since $\norm{\tilde{\mathsf{x}}}=R_\ell$ on ring $\ell$. Substituting $\tilde{y}=y$ yields~\eqref{eq:mmse_estimate} and~\eqref{eq:posterior_var}.
The \ac{MAP} estimate maximizes the mixed posterior~\eqref{eq:orbital_posterior} hierarchically.
As only the von Mises kernel depends on $\theta$ and $\kappa_\ell\ge 0$, the conditional mode is $\theta^\star=\angle\tilde{y}$ for every ring, selected by its posterior mass $\hat{r}_\ell$, and since the softmax normalizer is $\ell$-independent, $\arg\max_\ell\hat{r}_\ell = \arg\max_\ell\Lambda_\ell$. Substituting $\tilde{y}=y$ yields~\eqref{eq:map_estimate}.
\end{proof}

\vspace{-1ex}
Both detectors require $L$ Bessel evaluations against the $\mathcal{O}(M)$ of~\eqref{eq:exact_posterior}. The framework's sole approximation is operational, namely evaluating the orbital posterior of $\tilde{\mathsf{x}}$ at $\tilde{y}=y$, the realized observation from the true channel $\mathsf{y}=\mathsf{x}+\mathsf{z}$.

For the mixed posterior~\eqref{eq:orbital_posterior}, ``maximizing'' is meant hierarchically, in the sense that the ring is chosen by its probability \emph{mass} $\hat{r}_\ell$, which is invariant to phase reparametrization and Bayes-optimal for ring classification under 0--1 loss in the orbital model, and the phase by its conditional \emph{mode} $\angle y$. 
Consequently~\eqref{eq:map_estimate} is not the ordinary \ac{MAP} decision~\eqref{eq:map_exact} over $\mathcal{M}$, and it takes values in the relaxed support $\{R_\ell e^{j\theta}\}$, such that a final hard slicing onto $\mathcal{M}$ is required whenever a discrete decision is needed, exactly as for $\hat{x}_{\mathrm{OMMSE}}$ and the \ac{LMMSE} baseline.
Maximizing the joint \emph{density} value instead adds the von Mises peak height $\kappa_\ell - \ln I_0(\kappa_\ell)$ to $\Lambda_\ell$, and the two rules are compared next.

\begin{figure*}[htbp]
\vspace{0.02in}
\centering
\includegraphics[width=1.8\columnwidth]{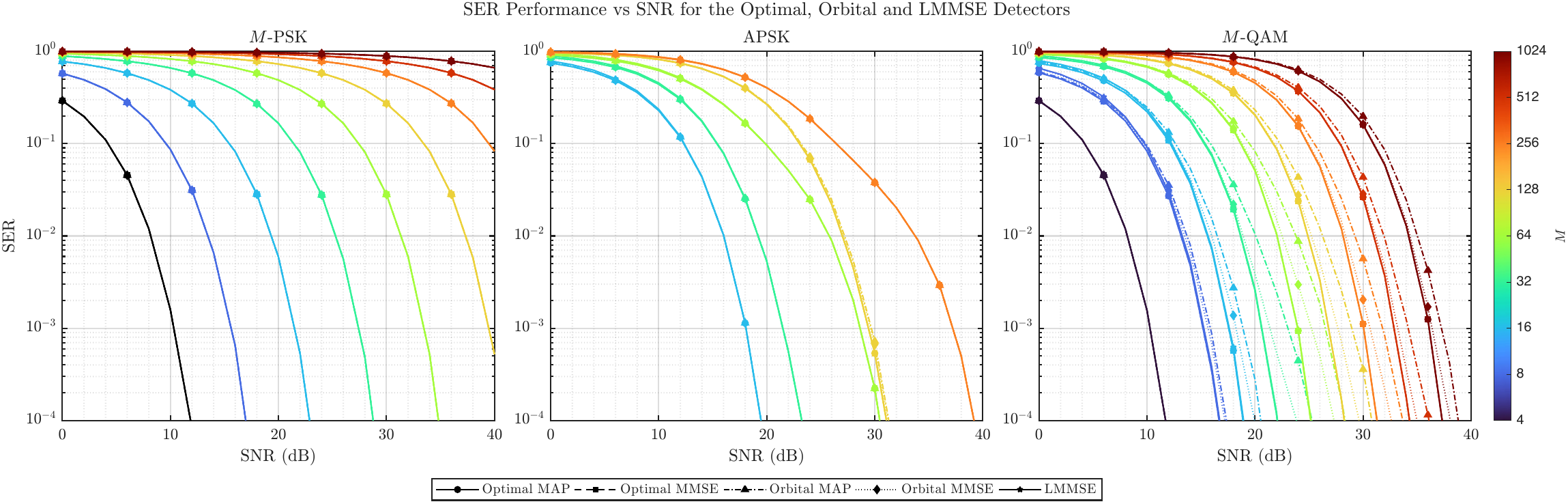}
\vspace{-1ex}
\caption{\ac{SER} performance of the orbital \ac{MMSE} and \ac{MAP} detectors compared with their exact discrete counterparts and the energy-matched \ac{LMMSE} baseline.}
\label{fig:SER}
\vspace{1ex}
\centering
\includegraphics[width=1.8\columnwidth]{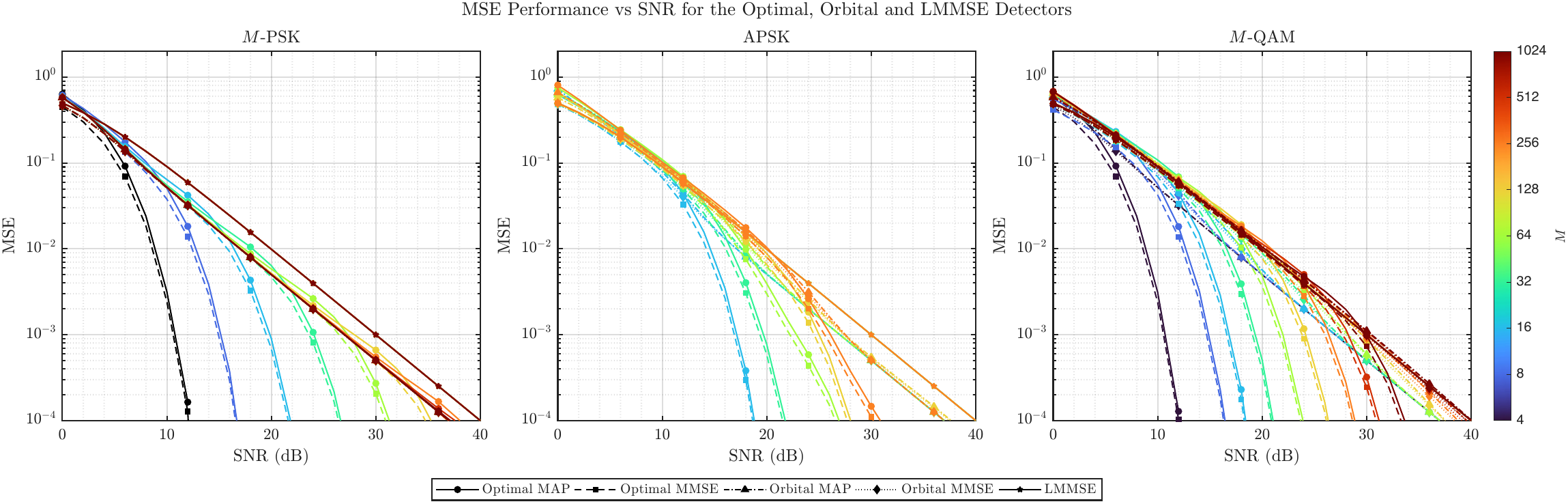}
\vspace{-1ex}
\caption{\ac{MSE} performance of the orbital \ac{MMSE} and \ac{MAP} detectors compared with their exact discrete counterparts and the energy-matched \ac{LMMSE} baseline.}
\label{fig:MSE}
\vspace{-3ex}
\end{figure*}
\vspace{-1ex}

\begin{proposition}[Ring boundary: Posterior mass versus joint density]
\label{prop:ring_boundary}
Fix two neighboring rings $\ell$ and $\ell+1$, with $0 < R_\ell < R_{\ell+1}$, and set $\rho \triangleq \norm{y}$ and $\Delta_\ell \triangleq R_{\ell+1} - R_\ell > 0$. 
Both the posterior-mass rule $\arg\max_\ell \Lambda_\ell$ of~\eqref{eq:map_estimate} and the joint posterior-density mode $\arg\max_\ell \big[\Lambda_\ell + \kappa_\ell - \ln I_0(\kappa_\ell)\big]$ select the outer ring $\ell+1$ over $\ell$ if and only if $\rho$ exceeds a threshold $\rho_\ell \in \mathbb{R}$, and these thresholds are given as
\vspace{-1ex}
\begin{subequations}
\label{eq:ring_boundaries}
\begin{align}
\label{eq:boundary_density}
\rho^{\mathrm{den}}_\ell &= \frac{R_\ell + R_{\ell+1}}{2} \;-\; \frac{\sigma_{\mathsf{z}}^{2}}{2\Delta_\ell}\,\ln\frac{r_{\ell+1}}{r_\ell},\\[0.3ex]
\label{eq:boundary_mass}
\rho^{\mathrm{mass}}_\ell &= \rho^{\mathrm{den}}_\ell
\;+\; \underbrace{\frac{\sigma_{\mathsf{z}}^{2}}{4\Delta_\ell}\,\ln\frac{R_{\ell+1}}{R_\ell}}_{\textstyle \triangleq\, \delta_\ell \,>\, 0}
\;+\; \mathcal{O}\big(\sigma_{\mathsf{z}}^{4}\big),
\end{align}
\end{subequations}
where~\eqref{eq:boundary_density} is exact and~\eqref{eq:boundary_mass} holds as $\mathsf{snr}\to\infty$.

Whenever $\rho^{\mathrm{den}}_\ell>0$, the strict ordering $\rho^{\mathrm{mass}}_\ell > \rho^{\mathrm{den}}_\ell$ moreover holds \emph{exactly}, at every $\sigma_{\mathsf{z}}^{2}>0$;
and whenever the density threshold is positive, the mass boundary is displaced \emph{outward} relative to the density boundary.
At high SNR, their separation is $\delta_\ell = \mathcal{O}(1/\mathsf{snr}) > 0$.
\end{proposition}

\vspace{-2ex}
\begin{proof}
Adding $\kappa_\ell - \ln I_0(\kappa_\ell)$ to $\Lambda_\ell$ of~\eqref{eq:ring_logmetric} cancels the Bessel term and, using $\kappa_\ell = 2R_\ell\rho/\sigma_{\mathsf{z}}^2$, gives the density metric $\ln r_\ell - (R_\ell-\rho)^2/\sigma_{\mathsf{z}}^2$ up to the $\ell$-independent $\rho^2/\sigma_{\mathsf{z}}^2$. 
Its increment across the pair has $\rho$-derivative $2\Delta_\ell/\sigma_{\mathsf{z}}^2 > 0$, so the rule is a threshold rule, and equating the two metrics and dividing by $2\Delta_\ell/\sigma_{\mathsf{z}}^2$ yields~\eqref{eq:boundary_density}.
For the mass rule, $\mathrm{d}\ln I_0(\kappa_\ell)/\mathrm{d}\rho = 2R_\ell A(\kappa_\ell)/\sigma_{\mathsf{z}}^2$, so $\Lambda_{\ell+1}-\Lambda_\ell$ has $\rho$-derivative $2[R_{\ell+1}A(\kappa_{\ell+1}) - R_\ell A(\kappa_\ell)]/\sigma_{\mathsf{z}}^2 > 0$ for $\rho>0$, since $R_{\ell+1}>R_\ell$ and $A$ is strictly increasing with $\kappa_{\ell+1}>\kappa_\ell$; the rule is therefore also a threshold rule.
Substituting the expansion $\ln I_0(\kappa) = \kappa - \tfrac{1}{2}\ln(2\pi\kappa) + \mathcal{O}(\kappa^{-1})$ into~\eqref{eq:ring_logmetric} gives
\vspace{-1ex}
\begin{equation}
\label{eq:mass_metric_hs}
\Lambda_\ell = \ln r_\ell - \frac{(R_\ell-\rho)^2}{\sigma_{\mathsf{z}}^{2}} - \frac{1}{2}\ln R_\ell + c(\rho) + \mathcal{O}(\kappa_\ell^{-1}),
\vspace{-1ex}
\end{equation}
with $c(\rho)$ independent of $\ell$.

Equating~\eqref{eq:mass_metric_hs} across the pair and dividing by $2\Delta_\ell/\sigma_{\mathsf{z}}^2$ gives~\eqref{eq:boundary_density} plus $\tfrac{\sigma_{\mathsf{z}}^2}{4\Delta_\ell}\ln(R_{\ell+1}/R_\ell)$, while the residual contributes $\tfrac{\sigma_{\mathsf{z}}^2}{2\Delta_\ell}\mathcal{O}(\kappa^{-1}) = \mathcal{O}(\sigma_{\mathsf{z}}^{4})$, because the root lies at $\rho = (R_\ell+R_{\ell+1})/2 + \mathcal{O}(\sigma_{\mathsf{z}}^{2})$, which is bounded away from zero, so that $\kappa^{-1} = \mathcal{O}(\sigma_{\mathsf{z}}^{2})$, yielding \eqref{eq:boundary_mass}.
Let $\varphi(\kappa) \triangleq \kappa - \ln I_0(\kappa)$, so that the density gap exceeds the mass gap by $\varphi(\kappa_{\ell+1}) - \varphi(\kappa_\ell)$, which is strictly positive for $\rho>0$ because $\varphi'(\kappa) = 1 - A(\kappa) > 0$ and $\kappa_{\ell+1}>\kappa_\ell$, and which vanishes at $\rho=0$.
If $\rho^{\mathrm{den}}_\ell>0$, the mass gap is therefore negative at $\rho=0$ and strictly negative at $\rho^{\mathrm{den}}_\ell$, where the density gap vanishes, such that its root obeys $\rho^{\mathrm{mass}}_\ell > \rho^{\mathrm{den}}_\ell$ for every $\sigma_{\mathsf{z}}^{2}>0$.
\end{proof}

\vspace{-1ex}
The origin of $\delta_\ell$ is geometric. 
Since $I_0(\kappa)\sim e^{\kappa}/\sqrt{2\pi\kappa}$, the conditional phase density~\eqref{eq:vm_phase} peaks at $e^{\kappa_\ell}/\big(2\pi I_0(\kappa_\ell)\big)\sim\sqrt{\kappa_\ell/2\pi}\propto\sqrt{R_\ell}$ because the channel concentrates a fixed unit of probability, spread over the wider outer ring into a correspondingly narrower angular window.
Comparing joint density peaks therefore weights ring $\ell$ by an additional high-\ac{SNR} factor $\sqrt{R_\ell}$ relative to its mass $\hat{r}_\ell$, which is the outer-ring preference; equivalently, by~\eqref{eq:mass_metric_hs}, $\Lambda_\ell$ is the density metric penalized by $-\tfrac{1}{2}\ln R_\ell$ up to $\ell$-independent terms, and $\delta_\ell$ is the leading-order outward boundary displacement associated with this difference.
The hierarchical rule~\eqref{eq:map_estimate} is thus not a convenience but a quantifiable correction of order $1/\mathsf{snr}$ of the same order as the prior tilt already carried by the boundary, with the two rules are non-interchangeable at moderate \ac{SNR}.

\begin{remark}[Relation to prior ring constructions]
\label{rem:prior_art}
The identity behind~\eqref{eq:mmse_estimate} is not new: on one ring it is the $I_1/I_0$ denoiser of~\cite[App.~A]{Hara2026}, in general the Type-I formula of~\cite[Sec.~IV-E]{Tanahashi2011}, whose density matches Definition~\ref{def:orbital_prior}. 
New are that Theorem~\ref{thm:maxent} derives that density rather than positing it, that the Bessel functions are kept exact, and that~\cite{Tanahashi2011} gives no counterpart to~\eqref{eq:map_estimate} or Proposition~\ref{prop:ring_boundary}, its saving being $M^{N_{\mathrm{t}}}\!\to\!M^{N_{\mathrm{t}}-1}$, still $\mathcal{O}(M)$ per symbol.
\end{remark}

\newpage
\vspace{-2ex}
\section{Numerical Results}
\label{sec:numerical}
\vspace{-1ex}
We evaluate the orbital detectors on the scalar \ac{AWGN} channel~\eqref{eq:system_model} for three constellation families spanning the range of ring structure, namely $M$-\ac{PSK} ($L=1$), \ac{APSK}~\cite{Thomas1974,deGaudenzi2006,AneddaTBC2016} (a few rings), and $M$-\ac{QAM} ($L \ll M$ rings). 
The orbital \ac{MMSE}/\ac{MAP} detectors of Proposition~\ref{prop:orbital_estimators} are compared against the exact discrete detectors of~\eqref{eq:mmse_exact} and \eqref{eq:map_exact} and the energy-matched \ac{LMMSE} (Gaussian-prior) baseline. The \ac{SNR} is $\mathsf{snr} = \sigma_{\mathsf{x}}^2/\sigma_{\mathsf{z}}^2$ and each point averages $1\times10^{6}$ Monte Carlo symbols.

What is counted is the \emph{posterior/estimator evaluation complexity, excluding any final hard-slicing step where applicable}, since only the exact \ac{MAP} rule~\eqref{eq:map_exact} returns a discrete symbol, the others needing a final decision that adds the same term to each. 
The exact detectors require $\mathcal{O}(M)$ evaluations and the \ac{LMMSE} $\mathcal{O}(1)$, against $\mathcal{O}(L)$ for the orbital ones: $L=1$ for $M$-\ac{PSK}, small for \ac{APSK}, well under $M$ for $M$-\ac{QAM}.

Fig.~\ref{fig:SER} showcases the \ac{SER}. 
Across all families and orders the orbital detectors are all but indistinguishable from the exact discrete detectors. For $M$-\ac{PSK} the curves coincide, since a single ring reduces detection to a nearest-phase decision, while on the multi-ring $M$-\ac{QAM}/\ac{APSK} constellations only a small high-\ac{SNR} gap remains, which is the signature of the lone approximation, namely the relaxed phase. 
The \ac{LMMSE} baseline tracks the same waterfall, so hard-decision performance is essentially preserved across all three.

Finally, Fig.~\ref{fig:MSE} shows the \ac{MSE}, where the price of the phase relaxation becomes visible. 
The exact \ac{MMSE}/\ac{MAP} estimators waterfall, denoising both amplitude and phase, whereas the orbital and \ac{LMMSE} estimators share a shallower, phase-limited decay, since neither sharpens the phase beyond the observed $\angle y$. 
The orbital \ac{MMSE} nonetheless sits below the \ac{LMMSE} on every family, including $M$-\ac{PSK}: the two estimators differ even at $L=1$, since $\hat{x}_{\mathrm{OMMSE}}$ of~\eqref{eq:mmse_estimate} exploits the known radius through $A(\kappa_1)$, whereas the linear estimator scales $y$. 
Once the ring is resolved, only the tangential error survives, so the orbital \ac{MSE} tends to $\sigma_{\mathsf{z}}^2/2$ against the $\sigma_{\mathsf{z}}^2$ of the \ac{LMMSE}, a factor of two on every family. 
The orbital-optimal gap is therefore exactly the discrete phase the orbital prior relaxes by construction, which is the deliberate trade made in exchange for a closed-form $\mathcal{O}(L)$ posterior.

\section{Conclusion}
\label{sec:conclusion}
The maximum-entropy, radial-marginal-preserving prior induces a closed-form posterior, a softmax over $L$ amplitudes times a von Mises phase, from which orbital \ac{MMSE} and \ac{MAP} detectors are derived at $\mathcal{O}(L)$ per symbol, matching the exact discrete \ac{SER} to within a small high-\ac{SNR} gap. 
The ring relaxation and Bessel-ratio denoiser are prior art~\cite{Tanahashi2011,Hara2026}; what this article adds is that they are not approximations of convenience, with the orbital prior being the unique amplitude-preserving, phase-noncommittal relaxation of the constellation, together with the exact ring metric and the hierarchical \ac{MAP} rule that follow from it, whose pairwise posterior-mass threshold is shifted outward relative to the joint posterior-density threshold by a quantifiable factor upto a leading order at high \ac{SNR}.

\balance
\bibliographystyle{IEEEtran}
\bibliography{references}

@ARTICLE{GuoWuShamaiVerdu_TIT_2011,
  author={Guo, Dongning and Wu, Yihong and Shitz, Shlomo S. and Verdú, Sergio},
  journal={IEEE Transactions on Information Theory}, 
  title={Estimation in Gaussian Noise: Properties of the Minimum Mean-Square Error}, 
  year={2011},
  volume={57},
  number={4},
  pages={2371-2385},
  doi={10.1109/TIT.2011.2111010}}

@ARTICLE{ranganVAMP2019,
  author={Rangan, Sundeep and Schniter, Philip and Fletcher, Alyson K.},
  journal={IEEE Transactions on Information Theory}, 
  title={Vector Approximate Message Passing}, 
  year={2019},
  volume={65},
  number={10},
  pages={6664-6684},
  doi={10.1109/TIT.2019.2916359}}

@article{Donoho2009,
  title={Message-passing algorithms for compressed sensing},
  author={Donoho, David L and Maleki, Arian and Montanari, Andrea},
  journal={Proceedings of the National Academy of Sciences},
  volume={106},
  number={45},
  pages={18914--18919},
  year={2009},
  publisher={National Academy of Sciences}
}

@INPROCEEDINGS{Rangan2011,
  author={Rangan, Sundeep},
  booktitle={2011 IEEE International Symposium on Information Theory Proceedings}, 
  title={Generalized approximate message passing for estimation with random linear mixing}, 
  year={2011},
  volume={},
  number={},
  pages={2168-2172},
  doi={10.1109/ISIT.2011.6033942}}

@article{Suresh2026,
  author  = {S. Suresh and A. Michon and C. Poulliat and M. Guillaud and C. Goursaud},
  title   = {Leveraging von Mises Message-Passing for Massive MIMO Detection},
  journal = {Proc. IEEE International Workshop on Signal Processing and Artificial Intelligence in Wireless Communications (SPAWC)},
  year    = {2026}
}

@ARTICLE{Thomas1974,
  author={Thomas, C. and Weidner, M. and Durrani, S.},
  journal={IEEE Transactions on Communications}, 
  title={Digital Amplitude-Phase Keying with M-Ary Alphabets}, 
  year={1974},
  volume={22},
  number={2},
  pages={168-180},
  doi={10.1109/TCOM.1974.1092165}}

@ARTICLE{Guo2005,
  author={Dongning Guo and Shamai, S. and Verdu, S.},
  journal={IEEE Transactions on Information Theory}, 
  title={Mutual information and minimum mean-square error in Gaussian channels}, 
  year={2005},
  volume={51},
  number={4},
  pages={1261-1282},
  doi={10.1109/TIT.2005.844072}}

@ARTICLE{deGaudenzi2006,
  author={De Gaudenzi, R. and Guillen i Fabregas, A. and Martinez, A.},
  journal={IEEE Transactions on Wireless Communications}, 
  title={Performance analysis of turbo-coded APSK modulations over nonlinear satellite channels}, 
  year={2006},
  volume={5},
  number={9},
  pages={2396-2407},
  doi={10.1109/TWC.2006.1687763}}

@ARTICLE{Bayati2011,
  author={Bayati, Mohsen and Montanari, Andrea},
  journal={IEEE Transactions on Information Theory}, 
  title={The Dynamics of Message Passing on Dense Graphs, with Applications to Compressed Sensing}, 
  year={2011},
  volume={57},
  number={2},
  pages={764-785},
  doi={10.1109/TIT.2010.2094817}}

@book{Abramowitz1965,
  title={Handbook of Mathematical Functions: With Formulas, Graphs, and Mathematical Tables},
  author={Abramowitz, M. and Stegun, I.A.},
  isbn={9780486612720},
  lccn={lc65012253},
  series={Applied mathematics series},
  url={https://books.google.nl/books?id=MtU8uP7XMvoC},
  year={1965},
  publisher={Dover Publications}
}

@book{mardia2009directional,
  title={Directional Statistics},
  author={Mardia, K.V. and Jupp, P.E.},
  isbn={9780470317815},
  series={Wiley Series in Probability and Statistics},
  url={https://books.google.nl/books?id=PTNiCm4Q-M0C},
  year={2009},
  publisher={Wiley}
}

@ARTICLE{AneddaTBC2016,
  author={Anedda, Matteo and Meloni, Alessio and Murroni, Maurizio},
  journal={IEEE Transactions on Broadcasting}, 
  title={64-APSK Constellation and Mapping Optimization for Satellite Broadcasting Using Genetic Algorithms}, 
  year={2016},
  volume={62},
  number={1},
  pages={1-9},
  doi={10.1109/TBC.2015.2470134}}

@book{DLMF,
  author    = {{National Institute of Standards and Technology}},
  title     = {NIST Digital Library of Mathematical Functions},
  editor    = {Olver, Frank W. J. and Lozier, Daniel W. and Boisvert, Ronald F. and Clark, Charles W.},
  publisher = {Cambridge University Press},
  address   = {Cambridge, UK},
  year      = {2010},
  url       = {https://dlmf.nist.gov/},
  note      = {Release 1.0.28 (or later), accessed: 2026-05-04}
}

@ARTICLE{Cespedes2014,
  author={Céspedes, Javier and Olmos, Pablo M. and Sánchez-Fernández, Matilde and Perez-Cruz, Fernando},
  journal={IEEE Transactions on Communications}, 
  title={Expectation Propagation Detection for High-Order High-Dimensional MIMO Systems}, 
  year={2014},
  volume={62},
  number={8},
  pages={2840-2849},
  doi={10.1109/TCOMM.2014.2332349}}

@book{CoverThomas2006,
  author    = {T. M. Cover and J. A. Thomas},
  title     = {Elements of Information Theory},
  edition   = {2nd},
  publisher = {Wiley-Interscience},
  year      = {2006},
}

@book{Watson1944,
  author    = {G. N. Watson},
  title     = {A Treatise on the Theory of Bessel Functions},
  edition   = {2nd},
  publisher = {Cambridge University Press},
  year      = {1944}
}

@ARTICLE{YangHanzo2015,
  author={Yang, Shaoshi and Hanzo, Lajos},
  journal={IEEE Communications Surveys \& Tutorials}, 
  title={Fifty Years of MIMO Detection: The Road to Large-Scale MIMOs}, 
  year={2015},
  volume={17},
  number={4},
  pages={1941-1988},
  doi={10.1109/COMST.2015.2475242}}

@book{proakis2007digital,
  title={Digital Communications},
  author={Proakis, J.G. and Salehi, M.},
  isbn={9780071263788},
  lccn={2007036509},
  series={McGraw-Hill International Edition},
  url={https://books.google.de/books?id=ksh0GgAACAAJ},
  year={2008},
  publisher={McGraw-Hill}
}

@book{tse2005fundamentals,
  title={Fundamentals of Wireless Communication},
  author={Tse, D. and Viswanath, P.},
  isbn={9780521845274},
  lccn={2006272166},
  series={Wiley series in telecommunications},
  url={https://books.google.de/books?id=66XBb5tZX6EC},
  year={2005},
  publisher={Cambridge University Press}
}

@ARTICLE{ma2017orthogonal,
  author={Ma, Junjie and Ping, Li},
  journal={IEEE Access}, 
  title={Orthogonal AMP}, 
  year={2017},
  volume={5},
  number={},
  pages={2020-2033},
  doi={10.1109/ACCESS.2017.2653119}}

@article{Jaynes1957,
  title = {Information Theory and Statistical Mechanics},
  author = {Jaynes, E. T.},
  journal = {Phys. Rev.},
  volume = {106},
  issue = {4},
  pages = {620--630},
  numpages = {0},
  year = {1957},
  month = {May},
  publisher = {American Physical Society},
  doi = {10.1103/PhysRev.106.620},
  url = {https://link.aps.org/doi/10.1103/PhysRev.106.620}
}

@ARTICLE{WuVerd2010,
  author={Wu, Yihong and Verdú, Sergio},
  journal={IEEE Transactions on Information Theory}, 
  title={Rényi Information Dimension: Fundamental Limits of Almost Lossless Analog Compression}, 
  year={2010},
  volume={56},
  number={8},
  pages={3721-3748},
  doi={10.1109/TIT.2010.2050803}}

@article{Renyi1959,
  title={On the dimension and entropy of probability distributions},
  author={Alfr{\'e}d R{\'e}nyi},
  journal={Acta Mathematica Academiae Scientiarum Hungarica},
  year={1959},
  volume={10},
  pages={193-215},
  url={https://api.semanticscholar.org/CorpusID:121006720}
}

@ARTICLE{Tanahashi2011,
  author  = {Tanahashi, Makoto and Ochiai, Hideki},
  journal = {IEEE Trans. Signal Process.},
  title   = {A New Reduced-Complexity Conditional-Mean Based {MIMO} Signal Detection Using Symbol Distribution Approximation Technique},
  year    = {2011},
  volume  = {59},
  number  = {11},
  pages   = {5644--5651},
  doi     = {10.1109/TSP.2011.2163064}
}

@ARTICLE{Hara2026,
  author  = {Hara, Shoma and Takahashi, Takumi and Iimori, Hiroki and Ochiai, Hideki and Larsson, Erik G.},
  journal = {IEEE Trans. Wireless Commun.},
  title   = {Reciprocity Calibration of Dual-Antenna Repeaters via {MMSE} Estimation},
  year    = {2026},
  volume  = {25},
  pages   = {21244--21260},
  doi     = {10.1109/TWC.2025.3590790}
}

@article{ranasinghe2026orbital,
  title={Orbital Detection},
  author={Ranasinghe, Kuranage Roche Rayan and de Abreu, Giuseppe Thadeu Freitas},
  journal={arXiv preprint arXiv:2608.09362},
  year={2026}
}

\end{document}